\documentclass[a4paper,11pt,reqno]{article}
\usepackage[hmargin=1.in,vmargin=1in]{geometry}
\usepackage{amsthm}
\usepackage{amssymb}
\usepackage{amsmath}
\usepackage[dvips,pdfstartview=FitH,backref,pdfpagemode=None,colorlinks=true,citecolor=darkblue,linkcolor=red]{hyperref} \usepackage{color}
\usepackage{stmaryrd}
\usepackage[disable]{todonotes}
\newcommand{\uline}[1]{\underline{#1}}
\newcounter{tmpc}
\newcounter{cA}
\newcounter{cB}
\newcounter{cC}
\newcounter{cD}
\newcounter{cE}
\newcommand{\stacs}[1]{} 

\newtheorem*{3XOR}{\textbf {The 3XOR Principle}}

 \newtheorem{theorem}{Theorem}
\newtheorem{lemma}[theorem]{Lemma}

\newtheorem{definition}{Definition}
 \newtheorem{proposition}[theorem]{Proposition}
 \newtheorem{corollary}[theorem]{Corollary}
\newtheorem{fact}[theorem]{Fact}
\newtheorem{remark}[theorem]{Remark}
\def\squareforqed{\hbox{\rlap{$\sqcap$}$\sqcup$}}
\def\qed{\ifmmode\squareforqed\else{\unskip\nobreak\hfil
\penalty50\hskip1em\null\nobreak\hfil{\tt QED}
\parfillskip=0pt\finalhyphendemerits=0\endgraf}\fi}

\newcommand{\ignore}[1]{}

\newcommand {\cd}{\cdot}

\newcommand {\mar}[1]{
\marginpar{
                                                  \begin{minipage}{60pt}
                                             \tiny \color{white}#1
                                                  \end{minipage}
                                                       }
}
\newcommand {\ind} {\noindent}

\newcommand {\para}[1] {\paragraph{#1}}

\DeclareMathAlphabet{\mathitbf}{OML}{cmm}{b}{it}

\newcommand{\PTIME}{\textbf{P}}     

\newcommand{\NP}{\textbf{NP}}
\newcommand{\coNP}{\textbf{coNP}}

\newcommand{\Or}{\lor}

\ifx\theorem\undefined

\newtheorem{theorem}{Theorem}[section]
\newtheorem{lemma}[theorem]{Lemma}
\newtheorem{proposition}[theorem]{Proposition}
\newtheorem{corollary}[theorem]{Corollary}
\newtheorem{definition}{Definition}[section]

\newenvironment{remark}{\HalfSpace\par\noindent{\bf Remark}:}{\HalfSpace}
\fi

\newenvironment{notation}{\QuadSpace\par\noindent{\bf Notation}:}{\HalfSpace}

\ifx\proof\undefined
\newenvironment{proof}{\QuadSpace\par\noindent{\bf Proof}:}{\EndProof\HalfSpace}
\fi

\newcommand{\QuadSpace}{\vspace{0.25\baselineskip}}
\newcommand{\HalfSpace}{\vspace{0.5\baselineskip}}

\newcommand{\EndProof}{ \hfill \vrule width 1ex height 1ex depth 0pt }

\def\RL{{\mbox{\rm R(lin)}}}
\def\RZ{{\mbox{\rm R$$(lin)}}}
\def\RZQ{{\mbox{\rm R$$(quad)}}}
\def\RQ{{\mbox{\rm R(quad)}}}

\newcommand{\TCZ}{{\rm TC}\ensuremath{^0}}

\definecolor{bluetxt}{rgb}{0,0,.5}
\definecolor{myred}{rgb}{0.6,0.0,0.1}
\definecolor{greentxt}{rgb}{0,.5,0}
\definecolor{redtxt}{rgb}{0.1,0.1,0.65}
\definecolor{purpletxt}{rgb}{0.6,0.1,0.7}
\definecolor{black}{rgb}{.0,.0,.0}
\definecolor{verydarkblue}{rgb}{.0,.0,.4}
\definecolor{darkblue}{rgb}{.0,.0,.55}
\definecolor{lightgray}{rgb}{.7,.7,.7}

\newlength{\defbaselineskip}
\author{Iddo Tzameret\footnote{Institute for Theoretical Computer Science, The Institute for Interdisciplinary Information Sciences (IIIS), Tsinghua University, Beijing \texttt{tzameret@tsinghua.edu.cn} ~~Supported in part by the National Basic Research Program of China Grant 2011CBA00300, 2011CBA00301, the National Natural Science Foundation of P.~R.~China; Grants 61033001, 61061130540, 61073174, 61373002.} \vspace{3pt}  \\ IIIS, Tsinghua University}
\date{}

\begin{document}

\title{Sparser Random 3-SAT Refutation Algorithms and the Interpolation Problem}
\maketitle

\begin{abstract}
We formalize a combinatorial principle, called \textit{the 3XOR principle}, due to Feige, Kim and Ofek \cite{FKO06}, as a family of  unsatisfiable propositional formulas  for which  refutations of small size in any propositional proof system that possesses the feasible interpolation property imply an efficient \emph{deterministic} refutation algorithm for random 3SAT with $ n $ variables and $ \Omega(n^{1.4}) $ clauses. Such small size refutations would improve the state of the art  (with respect to the clause density) efficient refutation algorithm, which works only for $\Omega(n^{1.5})$ many clauses \cite{FO07}.  

We demonstrate polynomial-size refutations  of the 3XOR principle in  resolution operating with disjunctions of   quadratic  equations with small integer coefficients, denoted  \RZQ ; this is a weak extension of cutting planes with small coefficients. We show that \RZQ\ is weakly automatizable iff \RZ\ is weakly automatizable, where \RZ\ is  similar to \RZQ\ but with linear instead of quadratic equations (introduced in \cite{RT07}).
This reduces the problem of  refuting random 3CNF with $ n $ variables and $ \Omega(n^{1.4}) $ clauses to the interpolation problem of \RQ\ and to the weak automatizability of \RL.
\end{abstract}

\section{Introduction}
In the well known \emph{random 3-SAT model} one usually considers a distribution on formulas in conjunctive normal form (CNF) with $ m $ clauses and three literals each, where each clause is chosen independently with repetitions out of all possible $ 2^3\cd {n \choose 3}$ clauses with $ n $ variables (cf.~\cite{Ach09:SAT-handbook}). The \emph{clause density} of such a 3CNF is $m/n$. When $ m $ is greater than $ cn $ for sufficiently large $ c $, that is, when the clause density is greater than $c$,  it is known (and easily proved for e.g.~$c\ge 5.2$) that with high probability a random 3CNF is unsatisfiable.

A \emph{refutation algorithm} for random $k$CNFs is an algorithm that receives a $ k $CNF (with sufficiently large clause density) and outputs either ``\texttt{unsatisfiable}'' or ``\texttt{don't know}''; if the algorithm answers ``\texttt{unsatisfiable}" then the $ k $CNF is required to be indeed unsatisfiable; moreover, the algorithm should output ``\texttt{unsatisfiable}" with high probability (namely,  with probability \(1-o(n)\) over the input $k$CNFs).

We can view the problem of determining the complexity of 
(deterministic) refutation algorithms as an average-case version of the \PTIME\ vs.~\coNP\  problem: a polynomial-time refutation algorithm for random $k$CNFs (for a small enough clause density) can be 
interpreted as showing that ``\PTIME\ = \coNP\  in the average-case''; while a polynomial-time \textit{nondeterministic}
refutation algorithm (again, for a small enough clause density) can be interpreted as ``\NP\ = \coNP\ in the average-case". 

Refutation algorithms for random $ k $CNFs were investigated in Goerdt and Krivelevich \cite{GK01} and subsequent works by Goerdt and Lanka \cite{GL03}, Friedman, Goerdt and Krivelevich \cite{FGK05}, Feige and Ofek \cite{FO07} and Feige \cite{Fei07} and \cite{COG+07} (among other works). For random 3CNFs, the best (with respect to the clause density) polynomial-time refutation algorithm to date works for formulas with  at least $ \Omega(n^{1.5}) $ clauses \cite{FO07}. On the other hand, Feige, Kim and Ofek \cite{FKO06} considered efficient \emph{nondeterministic} refutation algorithms; namely, short \emph{witnesses} for unsatisfiability of 3CNFs that can be checked for correctness in polynomial-time. They established the current best (again, with respect to the clause density) efficient, alas \textit{nondeterministic}, refutation procedure: they showed that with probability converging to $ 1 $ a random 3CNF with $ n $ variables and $\Omega(n^{1.4} )$ clauses has a witness of size polynomial in $n$.

Since the current state of the art random 3CNF refutation algorithm works for $\Omega(n^{1.5})$ clauses, while the best nondeterministic refutation algorithm works already for $O(n^{1.4})$,  determining whether a deterministic polynomial-time (or even a quasipolynomial-time) refutation algorithm for random 3CNFs with $ n $ variables and $\Omega(n^{1.4}) $ clauses exists is to a certain extent the frontier open problem in the area of efficient refutation algorithms.

\subsection{Results}

In this work we reduce the problem of devising an efficient deterministic refutation algorithm for random 3CNFs with $\Omega(n^{1.4}) $ clauses to the interpolation problem in  propositional proof complexity. For a refutation system $ \mathcal P$, the \textit{interpolation problem for $\mathcal{P}$} is the problem that asks, given a $\mathcal P$-refutation of an 
unsatisfiable formula $A(x,y)\land B(x,z),$ for $x,y,z$ mutually disjoint sets of variables, and an assignment $\alpha$ for $x$, to return $0$ or $1$, such that if the answer is $0$ then $A(\alpha,y)$ is unsatisfiable and if the answer is $1$ then $B(\alpha,z)$ is unsatisfiable. If the interpolation problem for a refutation system $\mathcal P$ is solvable in time $T(n)$ we say that $\mathcal P$ has \textit{interpolation in time \(T(n)\)}.\footnote{We 
do not distinguish in this paper between proofs and refutations: proof systems prove tautologies and refutation systems refute unsatisfiable formulas (or,  equivalently prove the negation of unsatisfiable formulas).}
%
%
When $ T(n) $ is a polynomial we say that \(\mathcal P\) has \textit{feasible interpolation}. The notion of feasible interpolation was proposed in \cite{Kra94-Lower} and developed further in 
\cite{Razb95-Unprovability,BPR97,Kra97-Interpolation}.

We present a family of unsatisfiable propositional formulas,
\textit{\textbf{denoted $\Upsilon_{n}$}} and called \emph{the 3XOR principle formulas}, expressing a combinatorial principle, such that for any given refutation system $\mathcal P$ that admits short refutations of $\Upsilon_{n}$, solving efficiently the interpolation problem for $\mathcal P$ provides an efficient \textit{deterministic} refutation algorithm for random 3CNFs with $ \Omega(n^{1.4}) $ clauses. In other words, we have the following:
\setcounter{cA}{\thetheorem}
\begin{theorem}\label{thm:main-in-intro}
If there exists a propositional proof system \(\mathcal P\) that has interpolation  in time \(T(n)\) and that admits \(s(n)\)-size refutations of $\Upsilon_n$, then there is a \textit{deterministic} refutation algorithm for random 3CNF formulas with \(n\) variables and $\Omega(n^{1.4})$ clauses that runs in  time \(T(s(n))\). In particular, if \(\mathcal P\) has feasible interpolation and admits polynomial-size refutations of \(\Upsilon_n\) then the refutation algorithm runs in polynomial-time.
\end{theorem}


The argument  is based on the following: we show that the \textit{computationally hard part} of the Feige, Kim and Ofek~nondeterministic refutation algorithm  (namely, the part we do not know how to efficiently compute deterministically) corresponds  to a disjoint \NP-pair. Informally, the pair $(\mathbf{A},\mathbf{B})$ of disjoint \NP\ sets is the following:  $\mathbf A$ is the set of 3CNFs that have  a certain combinatorial property, that is, they contain a collection of sufficiently many \emph{inconsistent even $k$-tuples}, as defined by Feige et al. (see Definition \ref{def:incons-even-k-tuple}); and  $\mathbf B$ is the set of 3CNFs with $m$ clauses for which there exists an assignment that satisfies more than $m-\ell$ clauses as 3XORs (for $\ell$ a certain function of the number of variables $n$).

Theorem \ref{thm:main-in-intro} then follows from the known relation between disjoint \NP-pairs and feasible interpolation
\cite{Razb94,Pud03-tcs}:
in short, if $\mathbf{A}$ and $\mathbf{B}$ are two disjoint \NP\ sets and  $A(x,y)$ and $B(x,z)$ are the two polynomial-size Boolean formulas corresponding to $\mathbf{A}$ and $\mathbf{B}$, respectively (i.e., for all $x$, there exists a short $y$ such that $A(x,y)=1$ iff $x\in \mathbf{A}$; and similarly for $\mathbf{B}$), then  short refutations of $A(x,y)\land B(x,z)$ imply a polynomial-size algorithm that separates $\mathbf{A}$ from $\mathbf{B}$. 
For more on the relation between disjoint \NP-pairs and propositional proof complexity see, e.g., \cite{Pud03-tcs,AB04}.

In general, we observe that every efficient refutation algorithm (deterministic or not) corresponds directly to a disjoint \NP-pair as follows:
%
%
every efficient refutation algorithm is based on some property $P$ of CNFs that can be witnessed (or better, found) in polynomial-time. Thus, every efficient refutation algorithm corresponds to a family of formulas  $P(x)\to $\,$\neg$SAT$(x)$, expressing that if the input CNF has the property $P$ then $x$ is unsatisfiable; thus, $P(x)$ and SAT$(x)$ are two disjoint \NP\ predicates. In the case of the refutation algorithm of Feige, Kim and Ofek, $P(x)$ expresses simply that the 3CNF $x$ has the Feige et al.~witness. \emph{However, the disjoint {\rm \NP}-pair $(\mathbf{A},\mathbf{B})$ we work with is not of this type}. Namely, $\mathbf{A}$ is  not the predicate $P(x)$ for the full Feige, Kim and Ofek witnesses, rather a specific combinatorial predicate (mentioned above) that is only one ingredient in the definition of the Feige et al.~witness; and $\mathbf{B}$ is not SAT$(x)$. This saves us the trouble to formalize and prove in a weak propositional proof system the full Feige et al.~argument (such a formalization was done recently in \cite{MT10}; see Sec.~\ref{sec:consequences} for a comparison with \cite{MT10}).
%
%
\smallskip 

In the second part of this paper (Section \ref{sec:encoding}  onwards) we reduce the problem of determinizing the Feige et al.~nondeterministic refutation algorithm to the interpolation problem of a concrete and apparently weak refutation system. Specifically, we demonstrate polynomial-size refutations for $\Upsilon_n$ in a refutation system denoted \RZQ\ that extends both the cutting planes with small coefficients\footnote{A refutation in \textit{cutting planes with small coefficients} is a restriction of cutting planes in which all intermediate inequalities are required to have coefficients bounded in size by a polynomial in $n$, where $n$ is the size of the formula to be refuted (see \cite{BPR97}).} (cf.~\cite{CCT87,BPR97,Pud97}) and Res(2) (for any natural $k$, the system Res($k$) is   resolution that operates with $k$DNFs instead of clauses, introduced by Kraj\'{i}\v{c}ek   \cite{Kra01-Fundamenta}). We note also that \RZQ\ is a subsystem of $\TCZ$-Frege. 

An \RZQ\ refutation  (see Section \ref{sec:define-proof-systems} for a formal definition) over the variables $\{x_1,\ldots,x_n\}$ operates with \emph{disjunctions} of quadratic equations, where each quadratic equation is of the form:
$$ \sum_{i,j\in [n]} c_{ij}x_{i}x_j + \sum_{i\in[n]}c_ix_i+c_0 =a, $$
in which all $c_i, c_{ij}$ and $a$ are integers written in unary representation. The system \RZQ\ has the following derivation rule, which can be viewed as a generalized resolution rule: from two disjunctions of quadratic equations $\bigvee_i L_i\lor (L=a)$ and  $\bigvee_j L_j \lor (L'=b)$ one can derive:\vspace{-7pt}
 $$\bigvee_i L_i \lor \bigvee_j L_j \lor (L-L'=a-b).$$

\vspace{-8pt}

\ind We also add axioms that force our variables to be $0,1$.
An \RQ\ refutation of an unsatisfiable set $S$ of disjunctions of quadratic equations is a sequence of disjunctions of quadratic equations (called \textit{proof-lines}) that terminates with $1=0,$ and such that every proof-line is either an axiom, or appears in $S,$ or is derived from previous lines by the derivation rules. 

We show the following:
\setcounter{cB}{\thetheorem}
\begin{theorem}\label{thm:RZQ-proofs}
\RZQ\ admits polynomial-size refutations of the 3XOR principle formulas $\Upsilon_n$.
\end{theorem}
This polynomial upper bound  on the refutation size of the 3XOR principle is non-trivial because the encoding of the  3XOR formula is  complicated in itself and further the refutation system is very restrictive.  

By Theorem \ref{thm:main-in-intro}, we get the reduction from determinizing Feige et al.~work to the interpolation problem for \RQ. In other words:
%
%
\setcounter{cC}{\thetheorem}
\begin{corollary}
If \RZQ\ has feasible interpolation then there is a deterministic polynomial-time refutation algorithm for random 3CNFs with $n$ variables and $\Omega(n^{1.4})$ clauses.
\end{corollary}

Next we reduce the problem of determinizing the Feige et al.~refutation algorithm to the weak automatizability of a weaker system than \RQ, namely \RL, as explained in what follows. 

The concept of automatizability, introduced by Bonet, Pitassi and Raz \cite{BPR00} (following the work of \cite{KP98}),  is central to proof-search algorithms. The \textit{proof-search problem} for a refutation system $\mathcal P$ asks, given an unsatisfiable formula $\tau$, to find a $\mathcal P$-refutation of $\tau$. A refutation system $\mathcal P$ is \emph{automatizable} if for any unsatisfiable $\tau$ the proof-search problem for $\mathcal P$ is solvable in time polynomial in the smallest $\mathcal P$-refutation of $\tau$ (and equivalently, if there exists a polynomial-time algorithm that on input 
$\tau$ and a number $m$ in unary, outputs a $\mathcal P$-refutation of $\tau$ of size at most $m$ in, case such a refutation exists). 
Following Atserias and Bonet \cite{AB04}, we say that a refutation system $\mathcal P$ is \emph{weakly automatizable} if there exists an automatizable  refutation system $\mathcal P'$ that polynomially simulates $\mathcal P$.
Note that if $\mathcal P$ is not automatizable, it does \textit{not} necessarily follow  that also  $\mathcal P'$ is not automatizable. Hence, from the perspective of proof-search algorithms, weak automatizability is a more natural notion than automatizability (see \cite{Pud03-tcs} on this).

In \cite{RT07}, the system \RZ\ was introduced which is similar to \RZQ, except that all  equations are \textit{linear} instead of quadratic. In other words, \RL\ is resolution over linear equations with small coefficients. We show the following:
\setcounter{cD}{\thetheorem}
\begin{theorem}
\RZQ\ is weakly automatizable iff \RZ\ is weakly automatizable.
\end{theorem}
The proof of this theorem follows a similar argument to Pudl\'ak \cite{Pud03-tcs}. Since weak automatizability of a proof system implies that the proof system has feasible interpolation \cite{BPR00,Pud03-tcs}, we obtain the following:
\setcounter{cE}{\thetheorem}
\begin{corollary}
If \RZ\ is weakly automatizable then there is a deterministic refutation algorithm for random 3CNFs with $n$ variables and $\Omega(n^{1.4})$ clauses.
\end{corollary}

\subsection{Consequences and relations to previous work}\label{sec:consequences} The key point of this work is the relation between constructing an efficient refutation algorithm for the clause density  $\Omega(n^{0.4})$ to proving upper bounds in weak enough propositional proof systems for the 3XOR principle (namely, proof systems possessing feasible interpolation); as well as establishing such upper bounds in relatively week proof systems.

There are two ways to view our results: either  as \textbf{(i)} proposing an approach to improve the current state of the art in refutation algorithms via proof complexity upper bounds; or conversely as \textbf{(ii)} providing a \textit{new kind} of important computational consequences that  will follow from feasible interpolation and weak automatizability of weak proof systems. Indeed, the consequence  that we provide is  of a different kind from the group of  important recently discovered  algorithmic-game-theoretic consequences  shown by Atserias and Maneva \cite{AM10}, Huang and Pitassi \cite{HP11} and Beckmann, Pudl\'{a}k and Thapen \cite{BPT14}. In what follows we explain these two views  in more details. \medskip         

\textbf{(i)} Our results show that by proving that \RQ\ has feasible interpolation or by demonstrating a short refutation of the 3XOR principle in some refutation system that admits feasible interpolation, one can advance the state of the art in refutation algorithms. We can hope that if feasible interpolation of \RQ\ does not hold, perhaps interpolation in quasipolynomial-time holds (either for \RQ\ or for any other system admitting short refutations of the 3XOR principle), which would already improve exponentially the running time of the current best deterministic refutation algorithm for 3CNFs with $\Omega(n^{1.4})$ clauses, since the current algorithm works in time $2^{O(n^{0.2}\log n)}$ \cite{FKO06}.

As mentioned above, \RQ\ is a common extension of Res(2) and cutting planes with small coefficients (though it is apparently not the weakest such common extension because already \RL\ polynomially simulates both Res(2) and cutting planes with small coefficients).   Whether Res(2) has feasible interpolation (let alone, interpolation in quasi-polynomial time) is  open and there are no conclusive evidences for or against it. Note that by Atserias and Bonet \cite{AB04}, Res(2) has feasible interpolation iff resolution is weakly automatizable. However this does not necessarily constitute a strong evidence against the feasible interpolation of Res(2), because the question of whether resolution is \textit{weakly} automatizable is itself open, and there is no strong evidence ruling out a positive answer to this question\footnote{It is known that, based on reasonable hardness assumptions from parameterized complexity, resolution is not \textit{automatizable} by Alekhnovich and Razborov 
\cite{AR08}, which is, as the name indicates, a stronger property than weak automatizability.}. Similarly, there are no strong evidences that rule out the possibility that cutting planes is weakly automatizable.
\smallskip 

\textbf{(ii)} Even if our suggested approach is not expected to lead to an improvement in refutation algorithms, it is still interesting in the following sense. The fact that \RZQ\ has short refutations of the 3XOR principle provides a new evidence that (weak extensions of)   Res(2) and cutting planes with small coefficients may not have feasible interpolation, or at least that it would be highly non-trivial to prove they do have feasible interpolation; the reason for this is that establishing the feasible interpolation for such proof systems would entail quite strong algorithmic consequences, namely, a highly non-trivial improvement in refutation algorithms. 
This algorithmic consequence adds to other recently discovered and important algorithmic-game-theoretic consequences that would follow from feasible interpolation of weak proof systems. 

Specifically, in recent years several groups of researchers discovered connections between feasible interpolation and weak automatizability of small depth Frege systems to certain game-theoretic algorithms: Atserias and Maneva \cite{AM10} showed that solving \textit{mean payoff games }is reducible to the weak automatizability of depth-2 Frege (equivalently, Res($n$)) systems and to the feasible interpolation of depth-3 Frege systems (actually, depth-3 Frege where the bottom fan-in of formulas is at most two). Subsequently, Huang and Pitassi \cite{HP11} showed that if depth-3 Frege system is weakly automatizable, then \textit{simple stochastic games }are solvable in polynomial time.
 Finally, Beckmann, Pudl\'{a}k and Thapen \cite{BPT14} showed that weak automatizability of resolution implies a polynomial-time algorithm for the \textit{parity game}.


\para{Comparison with  M\"uller and Tzameret \cite{MT10}.}\label{sec:compare-MT12}

In \cite{MT10} a polynomial-size \TCZ-Frege proof of the correctness of the Feige et al.~witnesses was shown. However the goal of \cite{MT10} was different from the current paper. In \cite{MT10} the goal was to construct short propositional refutations for random 3CNFs (with sufficiently low clause density). Accordingly, the connection to the interpolation  problem was not made in \cite{MT10}; and further, it is known by \cite{BPR00} that  \TCZ-Frege does not admit feasible interpolation (under cryptographical assumptions).  On the other hand, this paper aims to demonstrate that certain short refutations will  have \textit{algorithmic }consequences (for  refutation algorithms). Indeed, since we are not interested here to prove the correctness of the full Feige et al.~witnesses, we are isolating the computationally hard part of the witnesses from the easy (polytime computable) parts, and formalize the former part (i.e., the 3XOR principle) as a propositional formula in a way that is  suitable for the reduction to the interpolation problem. 

One advantage of this work over \cite{MT10} is that Theorem \ref{thm:RZQ-proofs} gives a more concrete logical characterization of parts of the Feige et al.~witnesses (because the proofs in \cite{MT10} were conducted indirectly, via a general translation from first-order proofs in  bounded arithmetic), and this characterization is possibly \textit{tighter} (because \RQ\ is 
apparently strictly weaker than \TCZ-Frege).  \section{Preliminaries}


Let $ F $ be a 3CNF with $n $ variables $ X =\{x_1,\ldots,x_n\}$ and $ m $ clauses.
We denote $ \{1,\ldots,n\}$ by $[n]$. The truth value of a formula $G$  under the Boolean assignment $ A $ is written $ G(A)$. An assignment $ A $ \textit{satisfies as a 3XOR} thew clause $\ell_1 \lor \ell_2 \lor \ell_3$
if $(\ell_1 \oplus \ell_2 \oplus \ell_3)(A)=1$ (where $\oplus $ denotes the XOR operation, and the $ \ell_i$'s are literals, namely variables or their negation).

\subsection{Disjoint NP-pairs and feasible interpolation of  propositional proofs}\label{sec:disj-NP-pair-recall}

In this section we review
the notion of a disjoint \NP-pair and its relation to propositional proofs and the feasible interpolation property.

A \textit{disjoint \NP-pair }is simply a pair of languages in \NP\ that are disjoint. Let \(L,N\) be a disjoint \NP-pair such that \(R(x,y)\) is the corresponding relation for \(L\) and \(Q(x,z)\) is the corresponding relation for \(N\); namely, there exists  polynomials $p,q$ such that  \(R(x,y)\) and \(Q(x,z)\) are polynomial-time relations where  \(x\in L\) iff \(\exists y, |y|\le p(|x|) \land  R(x,y)=\tt{true}\) and  \(x\in N\) iff \(\exists z, |z|\le q(|x|)\land  Q(x,z)=\tt{true}\).

Since both polynomial-time  relations \(R(x,y)\) and \(Q(x,z)\) can be converted into a family of polynomial-size Boolean circuits, they can be written as a family of polynomial-size (in $n$) CNF formulas (by adding extension variables, that we  may assume are incorporated in the certificates $y$ and $ z$).
Thus, let \(A_{n}(\overline x,\overline y)\) be a polynomial-size CNF in the variables  $\overline x =(x_1,\ldots,x_n )$ and $ \overline y =(y_1,\ldots,y_{\ell} ),$ that is true iff \(R(\overline x,\overline y)\) is true, and let \(B_{n}(\overline x,\overline z)\) be a polynomial-size CNF in the variables $\overline x$ and $\overline z=(z_1,\ldots, z_{m}),$ that is true iff \(Q(\overline x,\overline z)\) is true (for some $\ell, m$ that are polynomial in $n$). For every $n\in\mathbb{N}$, we define the following unsatisfiable CNF formula in three  mutually disjoint vectors of variables $\overline x, \overline y, \overline  z$:
\begin{equation}\label{eeeeq:A_and_B}
           ~~~~~~~    F_n:=A_{n}(\overline x ,\overline y)\land B_{n}(\overline x,\overline z).
\end{equation}
Note that because $\overline y$ and $\overline z$ are disjoint vectors of variables  and $A_{n}(\overline x ,\overline y)\land B_{n}(\overline x,\overline z)$ is unsatisfiable, it must be that given any $\overline x \in \{0,1\}^n$, either $A_{n}(\overline x ,\overline y)$ or $ B_{n}(\overline x,\overline z)$ is unsatisfiable (or both).
\mar{}

\para{Feasible interpolation.}
We use standard notions from the theory of propositional proof complexity (see \cite{BP98,Seg_BSL07,CK02,Kra95} for surveys and introductions to the field).
In particular, we sometimes mix between refutations (that is, proofs of unsatisfiability of a formula) and proofs (that is, proofs of  tautologies). From the perspective of proof complexity refutations of contradictions and proofs of tautologies are for most purposes the same.

A \emph{propositional proof system}
 $\mathcal  P$ is a polynomial-time relation $V(\pi,\tau)$ such that for every propositional formulas $\tau$ (encoded as binary strings in some natural way), $\tau$ is a tautology iff there exists a binary string $\pi$ (the supposed ``proof of $\tau$") with $V(\pi,\tau)=\texttt{true}$. (Note that $|\pi|$ is not necessarily polynomial in $|\tau|$.) A propositional proof system $\mathcal P$ \textit{polynomially-simulates} another propositional proof system $\mathcal Q$ if there is a polynomial-time
computable function $f$ that maps $\mathcal Q$-proofs to $\mathcal P$-proofs of the same tautologies.

Consider a family of unsatisfiable formulas $F_n:=A_{n}(\overline  x,\overline  y)\land B_{n}(\overline  x,\overline  z)$, $i\in\mathbb{N}$,  in mutually disjoint vectors of variables, as in (\ref{eeeeq:A_and_B}) above. We say that the Boolean function $f(\overline x)$ is \emph{the interpolant of $F_n$} if for every $n$ and every assignment $\overline \alpha$ to $\overline x$:
\begin{equation}\label{eq:interp-definition}
                \begin{array}{lll}
                                f(\overline  \alpha ) = 1 \quad & \Longrightarrow &\quad A_n(\overline  \alpha ,\overline  y)
                                \,\,\,{\rm{is}}\,{\rm{unsatisfiable}}; {\rm ~~and}\\
                                f(\overline  \alpha ) = 0 \quad & \Longrightarrow  &\quad B_n(\overline  \alpha ,\overline  z)
                                \,\,\,{\rm{is}}\,{\rm{unsatisfiable}}{\rm{.}}
                \end{array}
\end{equation}

In other words, if only $A_n(\overline \alpha, \overline y)$ is unsatisfiable (meaning that $B_n(\overline \alpha, \overline z)$ is satisfiable) then $f(\overline \alpha)=1,$ and if only $B_{n}(\overline \alpha, \overline z)$ is unsatisfiable (meaning that $A_{n}(\overline \alpha, \overline z)$ is satisfiable) then $f(\overline \alpha)=0$, and if both $A_{n}(\overline \alpha, \overline y)$ and $B_{n}(\overline \alpha, \overline z)$ are unsatisfiable then $f(\overline \alpha)$ can be either $0$ or $1$\stacs{\!}.
Note that  $L$ (as defined above) is precisely the set of those assignments $\overline  \alpha$
for which $A(\overline  \alpha, \overline  y)$ is satisfiable,
and $N$ is precisely the set of those assignments $\overline  \alpha$ for which
$B(\overline  \alpha, \overline  z)$ is satisfiable,
and $L$ and $N$ are disjoint by assumption, and so $f(\overline  x)$ \textit{separates} $L$ from $N$; namely, it outputs different values for those elements in $L$ and those elements in $N$.

\begin{definition}[Interpolation property]\label{def:interpol}
A propositional proof system $\mathcal P$  is said to have the \emph{interpolation property in time $T(n)$} if the existence of a size $s(n)$  $\mathcal P$-refutation of a family $F_n$ as in (\ref{eeeeq:A_and_B}) above implies the existence of an algorithm computing $f(\overline x)$ in $T(s(n))$ time.
When a proof system $\mathcal P$ has the interpolation property in time ${\rm poly}(n)$ we say that $\mathcal P$ has the \emph{feasible interpolation property}, or simply that $\mathcal P$ has \emph{feasible interpolation}.
\end{definition}

\todo{\tiny Put the above parag IN full VERSION!}

\subsection{Refutation algorithms} 
We repeat here the definition given in the introduction.
The distribution of \textit{random  3CNF formulas }with $n$ variables and $m$ clauses is defined by choosing $ m $ clauses with three literals each, where each clause is chosen independently with repetitions out of all possible $ 2^3\cd {n \choose 3}$ clauses with $ n $ variables. A \textit{refutation algorithm }for random 3CNFs is an algorithm $\mathsf A$ with input a 3CNF and two possible outputs ``\texttt{unsatisfiable}'' and ``\texttt{don't know}'', such that (i) if on input $C$, $\mathsf A$ outputs ``\texttt{unsatisfiable}'', then $C$ is unsatisfiable; and (ii) for any $n$, with probability at least $1-o(1)$ $\mathsf{A}$ outputs ``\texttt{unsatisfiable}'' (where the probability is considered over the distribution of random 3CNFs with $n$ variables and $m$ clauses, and where $o(1)$ stands for a term that converges to 0 when $n$ tends to infinity).
\section{The 3XOR principle}

The following definitions and proposition are due to Feige et al.~\cite{FKO06}.

\begin{definition}[Inconsistent even $ k $-tuple]\label{def:incons-even-k-tuple}
An \emph{even $ k $-tuple} is a tuple of $ k $ many $ 3 $-clauses in which every variable appears even times.
An \emph{inconsistent even $ k $-tuple} is an even $ k $-tuple in which the total number  of negative literals is odd.
\end{definition}
Note that for any even $ k$-tuple, $ k $ must be an even number (since by assumption the total number of variables occurrences $ 3k $ is even).
The following is the combinatorial principle, due to Feige et al.~\cite{FKO06} that we consider in this work:
\begin{3XOR}
Let $K$ be a 3CNF over the variables $X$. Let $ S $ be $ t $ inconsistent even $ k $-tuples from $K$, such that every clause from $ K $ appears in at most $ d $ inconsistent even $ k $-tuples in $ S $. Then, given any Boolean assignment to the variables $ X ,$ the number of clauses in $ K $ that are  unsatisfied by the assignment as 3XOR is at least $ \lceil t/d \rceil $.
\end{3XOR}

The correctness of the 3XOR principle follows directly from the following proposition and the fact that every clause in $ K $ appears in at most $ d $ even $k$-tuples in $ S $:

\begin{proposition}[\cite{FKO06}]\label{prop:3xor-princ}  For any inconsistent even $ k $-tuple (over the variables $X$) and any Boolean assignment $ A $ to $ X$, there must be a clause in the $ k $-tuple that is unsatisfied as 3XOR.
\end{proposition}
The proof follows a simple counting modulo 2. For completeness we prove this proposition.
\begin{proof}
Assume by a way of contradiction that for some assignment $ A $ every clause from the $ k $-tuple is satisfied as a 3XOR and recall that $ k $ must be even. Thus, if we sum modulo 2 all the literals in the $ k $-tuple \textit{via clauses}, then since $ k $ is even we get that the sum equals $ 0 $ modulo 2.

On the other hand, if we count \textit{via variables} then  summing modulo 2 all literals $ \ell_i(A) $ in the $k$-tuple, we get 1 (modulo 2), for the following reason. First, we sum all variables $ x_i $ that have odd number of negative occurrences. Because $ x_i $ appears an even number of times in the $k$-tuple, the number of positive occurrences of $ x_i $ is also odd. So in total all occurrences of $  x_i ( A) $ and \(\neg x_i(A)\) contribute $ 1 $ to our sum (modulo 2). There must be an odd number of such variables $ x_i $ in our $ k $-tuple because the $ k $-tuple is \emph{inconsistent}. Thus this sums up to $ 1  $ (modulo $2$).
Then we add to this sum those variables that have an even number of negative occurrences (and hence also an even number of positive occurrences); but they cancel out when summing their values under $ A $ modulo 2, and so they contribute 0 to the total sum. Hence, we get $ 1 $ as the total sum. This contradicts the  counting in the previous paragraph which turned out 0.
\end{proof}

\section{From short proofs to refutation algorithms}

In this section we demonstrate that polynomial-size proofs of (encodings of the) 3XOR principle in a proof system that has the feasible interpolation property yield  deterministic polynomial-time refutation algorithms for random 3CNF formulas with $\Omega(n^{1.4})$ clauses.

\subsection{The witness for unsatisfiability}
\label{sec:describe-the-FKO-witness}

Feige, Kim and Ofek nondeterministic refutation algorithm \cite{FKO06} is based on the existence of a polynomial-size witness of unsatisfiability for most 3CNF formulas with sufficiently large clause to variable ratio. The witness has several parts, but as already observed in \cite{FKO06},  apart from the $t$ inconsistent even $k$-tuples (Definition \ref{def:incons-even-k-tuple}), all the other parts of the witness  are known to be computable in polynomial-time.
In what follows we define the witnesses for unsatisfiability.

Let $K$ be a 3CNF with $n$ variables $ x_1,\ldots,x_n $ and $m$ clauses. The \emph{imbalance} of a variable $ x_i $ is the absolute value of the difference between the  number of its positive occurrences and the number of its negative occurrences. The \emph{imbalance of $K$} is the sum over the imbalances of all
 variables, in $K$, denoted $I(K)$. We define $M(K)$ to be an $ n \times n $ rational matrix $M$ as follows: let $i,j\in[n]$, and let $d$ be the number of clauses in $K$ where $x_i$ and $x_j$ appear with  different signs and $s$ be the number of clauses where $x_i$ and $x_j$ appear with the same sign.  Then $M_{ij}:=\frac{1}{2}(d-s)$. In other words, for each clause in $K$ in which $x_i$ and $x_j$ appear with the same sign we add $\frac{1}{2}$ to \(M_{ij}\) and  for each clause in $K$ in which $x_i$ and $x_j$ appear with different signs we subtract $\frac{1}{2}$ from \(M_{ij}\).
Let $\lambda $ be a rational approximation of the biggest  eigenvalue of $M(K)$. We shall assume that additive error of the approximation is $1/n^c$ for a constant $c$ independent of $n$; i.e., $|\lambda-\lambda'|\le 1/n^c$, for $\lambda'$ the biggest eigenvalue of $M(K)$; see \cite{MT10}.
%
%
\begin{definition}[FKO witness]\label{def:witness-with-3-items}
Given a 3CNF $K$, the \emph{FKO witness for the unsatisfiability of $K$} is defined to be the following collection:
\begin{enumerate}
\item the imbalance $I(K)$;
\item the matrix $M(K)$ and the (polynomially small) rational approximation \(\lambda\) of its largest eigenvalue;
\item \label{it:part-to-replace} a collection $S$ consisting of $t<n^2$ inconsistent even $k$-tuples such that every clause in $K$ appears in at most $d$ many even $k$-tuples, for some positive natural $k$;
\item \label{it:ineq-in-wit} the inequality  $t>\frac{d\cdot(I(K) +\lambda n)}{2} + o(1)$  holds.
\end{enumerate}
(The $ o(1) $ above stands for a specific rational number $b/n^c $, for $ c $ and $b$ constants independent of $n$).
\end{definition}

Feige et al.~\cite{FKO06} showed that if a 3CNF has a witness as above it is unsatisfiable. We have the following:
\begin{theorem}[\cite{FKO06}]\label{thm:FKO}
There are constants $c_0,c_{1}$ such that for a random 3CNF $K$ with $n$ variables and  $\Omega(n^{1.4})$ clauses, with probability converging to $1$ as $n$ tends to infinity there exist natural numbers
 $k,t,d$ such that $t=\Omega(n^{1.4})$ and
\vspace{-7pt} 
\begin{equation}\label{eq:fko-param-bounds}
~~~~~~~~~~~~~~~~~~~~k\le c_0\cd n^{0.2} \hbox{ ~and }  ~~~t<n^2 \hbox{ ~~and } ~~d\le c_1\cd n^{0.2}, 
\end{equation}

\vspace{-8pt} 
\ind and \(K\) has a witness for unsatisfiability as 
in Definition \ref{def:witness-with-3-items}.
\end{theorem}

%
%

Inspecting the argument in \cite{FKO06}, it is not hard to see that it is sufficient to replace part \ref{it:part-to-replace} in the witness with a witness for the following:

\begin{quote}
\textit{
\ref{it:part-to-replace}'. No assignment can satisfy more than $m-\lceil t/d\rceil-1$ clauses in $K$ as 3XORs. }
\end{quote}
Therefore, since  \(I(K)\), \(M(K)\) and \(\lambda\) are all polynomial-time computable (see \cite{FKO06} for this), in order to determinize the nondeterministic refutation algorithm of \cite{FKO06} it is sufficient to provide an algorithm that almost surely determines (correctly) that part \ref{it:part-to-replace}' above holds (when also $t$ and $d$ are  such that part \ref{it:ineq-in-wit} in the witness holds). In other words, in order to construct an efficient refutation algorithm for random 3CNFs (with $\Omega(n^{1.4})$ clauses) it is sufficient to have a deterministic algorithm $\mathsf A$ that on every input 3CNF (and for \(t\) and \(d\) such that part \ref{it:ineq-in-wit} in the witness holds)  answers either ``\texttt{condition \ref{it:part-to-replace}' is correct}'' or ``\texttt{don't know}'', such that $\mathsf A$ is never wrong (i.e., if it says ``\texttt{condition \ref{it:part-to-replace}' is correct}'' then condition \ref{it:part-to-replace}' holds) and with probability $1-o(n)$ over the input 3CNFs   $\mathsf A$ answers ``\texttt{condition \ref{it:part-to-replace}' is correct}''. Note that we do \underline{not} need to actually find the Feige et al.~witness \uline{nor} do we need to decide if it exists or not (it is possible that condition
\ref{it:part-to-replace}' holds but condition \ref{it:part-to-replace} does not, meaning that there is \emph{no} Feige et al.~witness).
%
%
The relation between unsatisfiability and bounding the number of clauses that can be satisfied as 3XOR in a 3CNF was introduced  by Feige in \cite{Fei02} (and used in \cite{FO07} as well as in \cite{FKO06}).


\subsection{The disjoint \NP-pair corresponding to the 3XOR principle}\label{sec:disj-NP-pair-definition}

We define the corresponding \emph{3XOR principle disjoint \NP-pair} as the pair of languages $(L,N)$, where $k,t,d$ are natural numbers given in \textit{unary}:

\vspace{-16pt}

\begin{multline*}
L:=\{\langle X,k,t,d \rangle \,\big|\; X \hbox{ is a 3CNF with $n$ variables and 
Equation (\ref{eq:fko-param-bounds}) holds for $k,t,d$}\\ ~~~~~~~~~~~~~~~~~~~~~~~~~~~\mbox{and there exists $t$ inconsistent
                 even $k$-tuples such that  } \\ \mbox{each clause of $X$ appears in no more than $d$ many $k$-tuples}\},
\end{multline*}

\vspace{-25pt}

\begin{multline*}
N:=\big \{\langle X,k,t,d \rangle \;\big |\; X \hbox{ is a 3CNF with $n$ variables and $m$ clauses and Equation (\ref{eq:fko-param-bounds})} \\ \mbox{~~~~~~~~~~~~~~~~~~~~~~~~~~holds for $k,t,d $ and there exists an assignment that }\\ \mbox{ satisfies at least  $m-\lceil t/d\rceil$ clauses in $X$ as 3XOR}\big\}.
\end{multline*}

It is easy to verify that both $ L $ and $N$ are indeed \NP\ sets, %
%
and that by the 3XOR principle, \(L\cap N=\emptyset\).

Using the same notation as in Section \ref{sec:disj-NP-pair-recall}, we denote by \(R_{}(x,y)\) and \(Q_{}(x,z)\) the polynomial-time relations for $L_{}$ and $N_{}$, respectively.
Further, for every $n\in\mathbb{N}$, there exists an \emph{unsatisfiable}  CNF formula in three  mutually disjoint sets of variables $\overline x, \overline y, \overline  z$:
\begin{equation}\label{eq:A_and_B}
                                \Upsilon_n:=A_{n}(\overline x ,\overline y)\land B_{n}(\overline x,\overline z),
\end{equation}
where \(A_{n}(\overline x,\overline y)\) and $B_{n}(\overline x,\overline z)$ are the CNF formulas expressing that $R_{}(x,y)$ and $Q_{}(x,z)$ are true for  $ x $ of length $n$, respectively.

\setcounter{tmpc}{\thetheorem}
\setcounter{theorem}{\thecA}
\begin{theorem}\label{cor:main}
Assume that there exists a propositional proof system
that has interpolation in time \(T(n)\) and that admits size $s(n)$
refutations of $ \Upsilon_n $.
Then, there is a \textit{deterministic} refutation algorithm for random 3CNF formulas with $\Omega(n^{1.4})$ clauses running in time \(T(s(n))\).
\end{theorem}
\setcounter{theorem}{\thetmpc}
\begin{remark}
Specifically, if the propositional proof system has feasible interpolation and admits polynomial-size refutations of $\Upsilon_n$ we obtain a polynomial-time refutation algorithm.
\end{remark}
\begin{proof}
By the assumption, and by the definition of the feasible interpolation property,  there exists a deterministic polynomial-time interpolant algorithm $\mathsf A_{}$ that on input a 3CNF \(K\) and three natural  numbers $k,t,d$ given in unary, if $\mathsf A(K,k,t,d)=1$ then $\langle K,k,t,d \rangle \not\in L_{}$ and if $\mathsf A_{}(K,k,t,d)=0$ then $\langle K,k,t,d \rangle \not\in N$.

The desired refutation algorithm  works as follows: it receives the 3CNF \(K\) and for  each 3-tuple of natural numbers $\langle k,t,d \rangle $ for which Equation (\ref{eq:fko-param-bounds}) holds it runs $\mathsf A(K,k,t,d)$. Note there are only \(O(n^{3})\) such 3-tuples.
If for one of these runs  $\mathsf A(K,k,t,d)=0$  then we know that $\langle K,k,t,d \rangle \not\in N$; in this case we check (in polynomial-time) that the inequality in Part \ref{it:ineq-in-wit} of the  FKO witness (Definition \ref{def:witness-with-3-items}) holds, and if it does, we answer ``\texttt{unsatisfiable}". Otherwise, we answer ``\texttt{don't know}".

The correctness of this algorithm stems from the following two points:

\ind\textbf{(i)} If we answered ``\texttt{unsatisfiable}", then there exist $k,t,d$ such that $\langle K,k,t,d \rangle \not\in N_{}$ and Part \ref{it:ineq-in-wit} in the FKO witness holds, and so Condition \ref{it:part-to-replace}' (from Section
\ref{sec:describe-the-FKO-witness}) is correct, and hence, by the discussion in \ref{sec:describe-the-FKO-witness},
$K$ is unsatisfiable.\smallskip

\ind\textbf{(ii)} For almost all 3CNFs we will answer ``\texttt{unsatisfiable}". This is because almost all of them will have an FKO witness (by Theorem \ref{thm:FKO}), which means that $\langle K,k,t,d \rangle\in L$ for some choice of  \(t<n^2,d,k\) (in the prescribed ranges) and hence the interpolant algorithm $\mathsf A$ must output 0 in at least one of these cases (because $\mathsf A(K,k,t,d)=1$ means that $\langle K,k,t,d \rangle\not\in L$).\end{proof}



\section{Short refutations of the 3XOR principle }
\label{sec:encoding}

In this section we define the propositional refutation system in which we  demonstrate polynomial-size refutations of the 3XOR principle. We then give an explicit encoding of the 3XOR principle as an unsatisfiable set of
disjunctions of  linear equations. 

\subsection{The propositional refutation systems \RZ\ and \RZQ}\label{sec:define-proof-systems}

The refutation system in which we shall prove the unsatisfiability of the 3XOR principle is denoted \RZQ. It is an extension of the refutation system \RZ\ introduced in \cite{RT07}.
The system \RL\ operates with disjunctions of linear equations with integer coefficients and \RQ\ operates with disjunctions of quadratic equations with integer coefficients, where in both cases the coefficients are written in unary representation.  We also add axioms that force all variables to be $0,1$. First we define the refutation system \RL. 

 The \textbf{size} of a linear equation $a_1 x_1+\ldots+ a_n x_n +a_{n+1}= a_0$
is defined to be $\sum_{i=0}^{n+1}{|a_i|}$, that is, the sum of the  sizes of all $a_i$ written in
\textit{unary} notation. The \emph{size of a disjunction of linear equations} is the total size of all linear equations in it. The \textbf{size }of a \textit{quadratic equation }and of a disjunction of quadratic equations is defined in a similar manner (now counting the size of  the constant coefficients, the coefficients of the linear terms and the coefficients of the quadratic terms). The \emph{empty disjunction} is unsatisfiable and stands for the truth value \textsf{false}.
\begin{notation} For $L$ a linear or quadratic sum and $S\subseteq \mathbb{Z}$, we write  $L\in S$, to denote the disjunction $\bigvee_{s\in S} L=s $. We call $L\in S$ a \emph{generalized linear (or quadratic) equation}.
\end{notation}

\begin{definition}[\RL]\label{def-R(lin)}
Let $K:= \{K_1,\ldots,K_m\}$ be a collection of disjunctions of linear
equations in the variables $ x_1,\ldots,x_n$. An \emph{R(lin)-proof from $K$ of a disjunction of linear equations
$D$} is a finite sequence $\pi =(D_1 ,\ldots,D_\ell)$ of disjunctions of linear
equations, such that $D_\ell=D$ and for every $i\in[\ell]$ one of the following holds:
\begin{enumerate}
\item  $\,D_i = K_j$, for some $j\in[m]$;

\item $D_i$ is a
\textbf{Boolean axiom} $\,x_t\in\{0,1\}$, for some $t\in[n]$;

\item $D_i$ was deduced by one of the following \RL-inference rules,
 using $D_j,D_k$ for some $j,k<i$:
\begin{description}

\item[\quad Resolution] Let $A,B$ be two, possibly empty, disjunctions of linear equations and let $L_1,L_2$ be two linear equations. From $A\Or L_1$ and $B\Or L_2$ derive $A \Or B \Or(L_1-L_2)$. (We assume that every linear form with $n$ variables is written as a sum of at most $n+1$ monomials.\footnote{Accordingly, in \RZQ\ we assume that every  quadratic sum with $n$ variables is written as a sum of at most $1+2n+ {n \choose 2}$ monomials.})

\item[\quad Weakening] From a possibly empty disjunction of linear equations $A$ derive
$A \Or L$\,, where $L$ is an arbitrary linear equation over the variables $x_1,\ldots,x_n$.

\item[\quad Simplification] From $A\Or (0=k)$ derive
$A$, where $A$ is a possibly empty disjunction of linear equations and $k\ne 0$.
\end{description}
\end{enumerate}
An \emph{R(lin) refutation of} a collection of disjunctions of linear equations
$K$ is a proof of the empty disjunction from $K$. The \textbf{size} of an R(lin) proof
$\pi$ is the total size of all the disjunctions of linear equations in
$\pi$ (where coefficients are written in unary representation).
\end{definition}

\begin{definition}[\RQ]\label{def-R(quad)}
The system \RQ\ is similar to \RL\ except that  proof-lines can be disjunctions of  quadratic equations with integer coefficients $\sum_{i,j}c_{ij}x_i x_j + \sum_i c_ix_i + c= S$ instead of  linear equations; and the \textbf{Boolean axioms} are now defined for all $i,j\in[n]$, as follows:
\[
                x_i\in\{0,1\}, ~~~~~~~~
                x_i+x_j-x_ix_j\in\{0,1\}, ~~~~~~~~~x_i-x_ix_j\in\{0,1\}\,.
\]
The size of an \RZQ\ refutation is the total size of all the proof-lines in it.
\end{definition}

Both \RZ\ and \RZQ\ can be proved to be sound and complete (for their respective languages, namely, disjunctions of linear and quadratic equations, respectively) refutation systems.

\smallskip

\subsection{Comparison of the refutation system \RZQ\ with other systems}\label{sec:choich_of_RQuad}

The \RZQ\ refutation  system is a weak propositional proof system that, loosely speaking, can both \emph{count} and \emph{compose mappings}, as we explain below.
 
Recall that the cutting planes refutation system with small coefficients operates with linear integer inequalities of the form $\sum_i a_i x_i \ge C$ (where the $a_i$'s are polynomial in the size of the formula to be refuted) that can be added, multiplied by a positive integer, simplified and divided by an integer $c$ in case $c$ divides every integer $a_i$, in which the division of the right hand side $C/c$ is rounded up (i.e., we obtain $\sum_i \frac{a_i}{c} x_i\ge \lceil \frac{C}{c} \rceil$). 

The cutting planes with small coefficients system can ``count" to a certain extent, namely it can prove efficiently certain unsatisfiable instances encoding counting arguments (like the pigeonhole principle). However, other simple counting arguments like the Tseitin graph formulas \cite{Tse68} are not known to have polynomial-size cutting planes refutations.

A weak extension of cutting planes with small coefficients is defined so to allow \textit{disjunctions }of linear equations (a big disjunction of linear equations can represent a single inequality). This way we obtain the system \RZ, that is similar to \RZQ\ but with \textit{linear} instead of quadratic equations. It was shown in \cite{RT07} that even when we allow disjunctions of only a \emph{constant number} of generalized\footnote{A 
\emph{generalized equation} is an equation $L\in S$, for 
$S\subset\mathbb{Z}$; which stands for the disjunction 
$\bigvee_{s\in S}L=s$.} linear equations in each proof-line,  \RZ\ has short refutations of the Tseitin formulas; this shows that using (fairly  restricted) disjunctions of linear equations allows to improve the ability of  cutting planes with small coefficients to refute contradictions that involve counting.

However, for our refutation of the 3XOR principle to work out we need to use quadratic instead of linear equations. Informally,  the reason for this is  to be able to ``compose maps": as observed by Pudl\'ak \cite{Pud03-tcs}, the reason why the \emph{$k$-Clique and $(k-1)$-Coloring} contradictions provably do not have short cutting planes refutations is that cutting planes cannot compose two mappings, which then makes it impossible to perform a routine reduction from the   \emph{$k$-Clique and $(k-1)$-Coloring} contradiction to the pigeonhole principle contradiction (and the latter contradiction does admit short cutting planes refutations). This is why Pudl\'ak introduced in \cite{Pud03-tcs} the system $CP^2$ which is cutting planes operating with \emph{quadratic inequalities}. The system \RZQ\ we work with is an extension of $CP^2$ (when the latter is restricted to small coefficients).



\subsection{The 3XOR principle formula}
We now describe the formula $\Upsilon_{n}$ encoding the 3XOR principle (the formula depends also on the parameters  $t,m $ and $k$, but we will suppress these subscripts). 

Recall that we wish to construct a family of formulas in three  mutually disjoint sets of variables $\overline X, \overline Y, \overline  Z$:
\begin{equation}\label{eq:Revisit-A_and_B}
                                \Upsilon_n:=A_{n}(\overline X ,\overline Y)\land B_{n}(\overline X,\overline Z),
\end{equation}
(where, in the terminology of Section \ref{sec:disj-NP-pair-definition}, \(A_{n}(\overline X,\overline Y)\) and $B_{n}(\overline X,\overline Z)$ are the CNF formulas expressing that $R_{}(x,y)$ and $Q_{}(x,z)$ are true for  $ x $ of length $n$, respectively).

%
%

Apart from the variables $\overline X,\overline Y, \overline Z$ we also add a group of variables, serving as extension 
variables: variables that encode the  product of two other variables, namely, (extension) variables that are forced to behave like products of two variables from $\overline X,\overline Y, \overline  Z$.
\textbf{\textit{We denote such extension variables with the $\llbracket \cd \rrbracket $ symbol; }}e.g., $\llbracket x_i\cd y_j \rrbracket $. 

Since we cannot use the $\overline Y$ variables in the second part of formula \ref{eq:Revisit-A_and_B} and we cannot use the $\overline Z$ variables in the first part of the formula \ref{eq:Revisit-A_and_B}, we can encode only products of variables from $\overline X,\overline Y$ and from $\overline X,\overline Z$, but \textit{not }products of a $\overline Y$ variable with a $\overline Z$ variable.
\smallskip 

It will be convenient sometimes to denote by  $x_{i+n}$  the literal $\neg x_i$, when it is assumed we use the $n$ variables $x_1,\ldots,x_n$ in the  3CNF encoded by $\overline X$.
\smallskip

\ind\textbf{A technical remark}: For the sake of simplicity  we \emph{do not} encode the three unary parameters $k,t,d$ (appearing in the disjoint \NP-pair in
Sec.~\ref{sec:disj-NP-pair-definition}) in our formula for $\Upsilon_n$ (and accordingly we do not encode the constraints in Equation (\ref{eq:fko-param-bounds})). This  slightly simplifies things, and does not harm the validity of the results, as it is easy to add these constraints to the formula and give short \RZQ\ refutations for such a  formulation.

\para{Variables and their meaning.}
The variables $\overline X$ correspond to the input 3CNF with $n$ variables. The variables $\overline Y$ correspond to the collections of $t$ many inconsistent even $k$-tuples. The $\overline Z=\{z_1,\ldots,z_n\}$ variables stand for a Boolean assignment
for the $n$ variables of the 3CNF.
(Note that we use the variables $x_i$ for the variables in the 3CNF and the variables $x_{ij}$ for the variables in our encoding of the 3CNF.)

The input 3CNF  $\overline X$ is encoded as a \(3m\times 2n\) table $\overline X$, where each block of three rows corresponds to a clause, and columns from 1 to $n$ correspond to positive literals occurrences, and columns $n+1$ to $2n$ correspond to negative literals occurrences. Formally, let $1\le i=3\cd l + r \le 3m$, where $r\in\{0,1,2\}, l\in[n],$ and  let $j\in[2n]$. Then $x_{ij}=1$ means that the $r$th literal in the $l$th clause in the input 3CNF is:

\begin{center} $x_j$ if $j\le n$, and  $\neg x_{j-n}$, if $j>n$.
\end{center}


The collection of $t$ inconsistent $k$ even tuples is encoded as $t$ tables, each table is encoded by the variables $\overline Y^{(s)}$, for $s\in[t]$. Each $\overline Y^{(s)}$ represents a table of dimension $k\times m$, where $y^{(s)}_{jl}=1 $ iff the $j$th member in the $s$th $k$-tuple is the $l$th clause (meaning the $l$th clause in the input 3CNF encoded by $\overline X$).

\paragraph{Group I of formulas (containing only $\overline X,\overline Y$):}

\begin{enumerate}
\item Every row in $\overline X$ contains exactly one $1$:
$$\sum_{j=1}^{2n} x_{ij} = 1\hbox{, ~~~for every }i\in[3m].$$



\item Every row in $\overline Y^{(s)}$ contains exactly one $1$: $$\sum_{j=1}^m y^{(s)}_{ij} =1\hbox{,~~~ for all }s\in[t], i\in[k].$$

\item Every column in $\overline Y^{(s)}$ contains at most one $1$:
$$\sum_{i=1}^k y^{(s)}_{ij} \in\{0,1\}\hbox{, ~~~for all }s\in[t], j\in[m].$$


\item \label{it:prod-variable} For any $ j\in[k],r\in[m],s\in[t],\ell\in[3m], i\in[2n]$, we introduce the new \textbf{\uline{\textbf{single} formal variable}}  $\llbracket y^{(s)}_{jr}\cd x_{\ell i}\rrbracket $ which will stand for the \emph{product} of two other formal variables
$y^{(s)}_{jr}\cd x_{\ell i}\,.$
For this we shall have the following axioms:
\[
          y^{(s)}_{jr} -\llbracket y^{(s)}_{jr}\cd x_{\ell i}\rrbracket \in \{0,1\}
\hbox{~~~and~~~}      x_{\ell} -\llbracket y^{(s)}_{jr}\cd x_{\ell i}\rrbracket \in \{0,1\}
\]  and
\[
                 y^{(s)}_{jr} + x_{\ell i} -\llbracket y^{(s)}_{jr}\cd x_{\ell i}\rrbracket \in\{0,1\}
                 \]

\setcounter{tmpc}{\theenumi}
\end{enumerate}

\ind
 As an abbreviation (\emph{not} a formal variable) we define the following:
\[
                Q_{ijh}^{(s)}:=\sum_{r=1}^m \llbracket y^{(s)}_{jr}\cd x_{(3(r-1)+h)i}\rrbracket\,, \hbox{~~~~~for all~} i=[2n] \hbox{~and~} h\in\{0,1,2\} \hbox{~and~} s\in[t],
\]
which expresses that $x_i$ occurs as the $h$th literal in the  $j$th clause of $\overline Y^{(s)}$.
\begin{enumerate}
\setcounter{enumi}{\thetmpc}


\item \label{it:007} We express that all the $\overline Y^{(s)}$'s are \emph{even} $k$-tuples (that is, that every variable $x_i$ appears even times) by:
\[
                                \sum_{r\in[k], h=0,1,2} Q_{irh}^{(s)}+Q^{(s)}_{(i+n)rh} \in \{0,2,4,\ldots,k\},~~~~~~~~\hbox{for all $i\in[n], s\in[t]$}
.\]
We can assume that $k$ is even, since for every even $k$-tuple $k$ must be even.

\item \label{it:encode-inconsistent} Similarly, we encode that the $\overline Y^{(s)}$'s are \emph{inconsistent} (that is, the number of negative literals in them is odd) by:
\[
  \sum_{r\in[k],h=0,1,2 \atop i\in[n]} Q^{(s)}_{(i+n)rh} \in \{0,3,5,\ldots,k-1\}.
\]
\item \label{it:only-d-joins} Every clause $i\in[m]$ appears in at most $d$ even $k$-tuples $\overline Y^{(1)},\ldots,\overline Y^{(t)}$. We put:
\[
                \sum_{j\in[k],  s\in[t]} y^{(s)}_{ji} \in \{0,1,\ldots,d\}, \hbox{~~~~~for every $i\in[m]$}.
\]

\setcounter{tmpc}{\theenumi}
\end{enumerate}
\medskip

\ind This finishes the encoding of the $t$ inconsistent even $k$-tuples.


\para{Group II of formulas (containing only $\overline X,\overline Z$):}

We now turn to the formulas expressing that there are assignments $\overline Z$ that satisfy more than  $m-\lceil t/d \rceil$ clauses in $\overline X$ as 3XORs. For every $j\in[3m],i\in[2n],\ell\in[n]$, let $\llbracket x_{ji}\cd z_\ell \rrbracket $ be a new formal variable that stands for the product $x_{ji}\cd z_\ell$. As in part \ref{it:prod-variable} of the formula above, we include the axioms that force $\llbracket x_{ji}\cd z_\ell \rrbracket $  to stand for $x_{ji}\cd z_\ell$.

Let us use the following abbreviation:
\[
U_j:=\sum_{h=0,1,2}\left(\sum_{i=1}^{n} \left\llbracket x_{(3(j-1)+h)i}\cd z_i \right\rrbracket  + \sum_{i=1}^{n} \left(x_{(3(j-1)+h)(i+n)}- \left\llbracket x_{(3(j-1)+h)(i+n)}\cd z_i \right\rrbracket\right)\right).
\]

Then, $U_j \in \{1,3\}$ states that the $j$th clause in $\overline X$ is satisfied as 3XOR by $\overline Z$. Note that  $x_{(3(j-1)+h)(i+n)}- \left\llbracket x_{(3(j-1)+h)(i+n)}\cd z_i \right\rrbracket$ is a linear term that expresses the quadratic term $x_{(3(j-1)+h)(i+n)}\cd(1-z_i)$.

\begin{enumerate}
\setcounter{enumi}{\thetmpc}
\item Let $u_j$ be a new formal variable expressing that the $j$th clause in $\overline X$ is satisfied as 3XOR by $\overline Z$. Hence,  $U_j \in\{1,3\}$ iff $ u_j=1 $, and we encode it as:
\[
 U_j \in\{0,2\} \lor (u_j=1) \hbox{~~and~~} U_j\in\{1,3\}\lor (u_j=0),
\]
\item \label{it:u-sum-i} There are assignments $\overline Z$ that satisfy more than  $m-\lceil t/d \rceil$ clauses in $\overline X$ as 3XORs:
\[
                \sum_{j=1}^m u_j\in\{m-\lceil t/d\rceil+1,\ldots,m\}.
\]
\end{enumerate}

The set of formulas described in this section has no $0,1$ solution by virtue of the 3XOR principle itself (Section \ref{sec:describe-the-FKO-witness}).



\section{Short refutations for the 3XOR principle}\label{sec:the-proof}
In this section we demonstrate polynomial-size (in $n$) \RZQ\ refutations of the 3XOR principle as encoded by disjunctions of linear equations in the previous section. 

\setcounter{tmpc}{\thetheorem}
\setcounter{theorem}{\thecB}
\begin{theorem}\label{thm:RZQ-proofs}
\RZQ\ admits polynomial-size refutations of the 3XOR principle formulas.
\end{theorem}
\setcounter{theorem}{\thetmpc}

We sometimes give only a high level description of the derivations. We use the terminology and abbreviations in Section \ref{sec:encoding}. We also use freely the ability of \RZ\ (and hence \RZQ) to count. For a detailed treatment of efficient counting arguments inside \RZ\ see \cite{RT07}.

\para{Step 1:}
Working in \RZQ, we  first show that our axioms prove that $\overline Z$ cannot satisfy as 3XOR all clauses of $\overline Y^{(s)}$, for any $s\in[t]$.

Recall from Section \ref{sec:encoding} the abbreviation
\[
                Q_{ijh}^{(s)}:=\sum_{r=1}^m \llbracket y^{(s)}_{jr}\cd x_{(3(r-1)+h)i}\rrbracket\,, \hbox{~~~~~for all~} i=[2n] \hbox{~and~} h\in\{0,1,2\} \hbox{~and~} s\in[t],
\]
 which stands for the statement  that $x_i$ occurs as the $h$th literal in the  $j$th clause of $\overline Y^{(s)}$ (and where $x_i$ for $i>n$ stands for the literal  $\neg x_{i-n}$). Let us use the  abbreviation:
$$ P_{jhs}:= \sum_{i=1}^{n} Q_{ijh}^{(s)}\cd z_i + \sum_{i=1}^{n} Q_{(i+n)jh}^{(s)}\cd (1-z_i).$$
Then, $P_{jhs}$ is a quadratic sum that stands for the statement that the $h$th literal in clause $j$ in $\overline Y^{(s)}$ is true under $\overline Z$.
Thus,
\begin{equation}\label{eq:three-p}
P_{j0s}+P_{j1s}+P_{j2s}\in\{1,3\}, \hbox{~~~for all $j\in[k]$}
\end{equation}
expresses that all the clauses in $\overline Y^{(s)}$ are satisfied as 3XOR under $\overline Z$. \smallskip

\emph{Our goal now is to  refute (\ref{eq:three-p}), based on our axioms}.
Informally, this refutation is done by counting: first count  by clauses in $\overline Y^{(s)}$, namely, add all left hand sides of (\ref{eq:three-p}) together reaching an even number (in the right hand side) by virtue of $k$ being even (recall we can assume that $k$ is even). Then, count by literals, namely sum all values of literals in $Y^{(s)}$ under the assignment $\overline Z$, which we can prove is odd from our axioms.
We now describe this refutation more formally.

Since $k$ is even, counting by clauses in $\overline Y^{(s)}$, namely, adding the left hand sides of (\ref{eq:three-p}) gives us easily the following (with a polynomial-size \RZQ\ proof):
\begin{equation}\label{eq:big-sum}
           \sum_{j=1}^k P_{j0s}+P_{j1s}+P_{j2s}\in\{0,2,4,\ldots,3k\}.
\end{equation}
Now we need to count by literals in $\overline Y^{(s)}$. We can abbreviate the number of occurrences in $\overline Y^{(s)}$ of the literal $x_i$, for $i\in[2n], s\in[t],$ by:
\[
                T_i:=\sum_{j\in[k]\atop h=0,1,2} Q_{ijh}^{(s)}\,.
\]
Let us abbreviate by $S_i$ the contribution of the literals $x_i$ and $\neg x_i$ to the total sum (\ref{eq:big-sum}). Thus
\[
                S_i:=  \sum_{j\in[k]\atop h=0,1,2} Q_{ijh}^{(s)}\cd z_i  + \sum_{j\in[k]\atop h=0,1,2} Q_{(i+n)jh}^{(s)}\cd (1-z_i).
 \]

It is possible to prove the following:
\begin{equation}\label{eq:T-or_S}
                T_i\in\{0,2,4,\ldots,k\}\lor S_i\in\{1,3,5,\ldots,k-1\}
\end{equation}
which states that if the number of occurrences in $\overline Y^{(s)}$ of the literal $x_i$ is odd then (since by our axioms stating that every variable occurs even times, the number of occurrences of the literal $\neg x_i$ must also be odd) the contribution of $x_i$ and $\neg x_i$ to the total sum (\ref{eq:big-sum}) is also odd (because either $z_i=0$ or $z_i=1$).

By the axioms saying that the number of negative literals is odd (axiom \ref{it:encode-inconsistent}) we get that:
\begin{equation}\label{eq:003}
                \sum_{i=1}^n T_{i+n}\in\{1,3,5,\ldots, k\cd n-1\}.
\end{equation}
And from the axioms stating that each variable occurs even times in $\overline Y^{(s)}$ we have:
\begin{equation}\label{eq:005}
                T_i+T_{i+n}\in\{0,2,4,\ldots,k\}, \hbox{~~for all $i\in[n]$}.
\end{equation}
From (\ref{eq:005}) we obtain
$
\sum_{i=1}^{2n} T_i \in \{0,2,4,\ldots,k\cd n\}
$,
and from this and (\ref{eq:003}) we obtain
\begin{equation}\label{eq:007}
\sum_{i=1}^n T_i \in\{1,3,5,\ldots,k\cd n-1\}.
\end{equation}
Note that  (\ref{eq:T-or_S}) can be interpreted as saying that if $T_i$ is odd then so does $S_i$. Accordingly, one can use (\ref{eq:T-or_S}) to substitute all $T_1,\ldots,T_n$ in (\ref{eq:007}) with $S_1,\ldots,S_n$, respectively.
We thus get that the total sum in the left hand side of (\ref{eq:big-sum}) is in $\{1,3,5,\ldots\}$, and we obtain a contradiction
with (\ref{eq:big-sum}). 


\bigskip

From a refutation of the collection of disjunctions (\ref{eq:three-p}), for any $s\in[t]$,
we can actually get the negation of this collection, that is:
\begin{equation}\label{eq:011}
  \bigvee_{j\in[k]} (P_{j0s}+P_{j1s}+P_{j2s})\in\{0,2\}.
\end{equation}
This stems from the following: it is already true in resolution that if we have a size $\gamma$ resolution refutation of $A_1,\ldots,A_l$, then assuming the axioms $A_1\lor B_1,\ldots,A_l\lor B_l$, we can have a size $O(\gamma\cd d)$ resolution derivation of $B_1\lor\ldots\lor B_l$, given that the total size of the $B_i$'s is $d$. To see this, take the resolution refutation of $A_1,\ldots,A_l$ and OR every line in this refutation with $B_1\lor\ldots\lor B_l$ (note that the resulting new axioms are actually derivable from the axioms $A_i\lor B_i$ via Weakening). Now, to get (\ref{eq:app_011}) from (\ref{eq:app_three-p}), we do the same, putting $P_{j0s}+P_{j1s}+P_{j2s}\in\{1,3\}$ instead of $A_j$ and $P_{j0s}+P_{j1s}+P_{j2s}\in\{0,2\}$ instead of $B_j$, for all $j\in[k]$,  noting that:
\begin{equation}\label{eq:app_three-p-two}
\left(P_{j0s}+P_{j1s}+P_{j2s}\in\{1,3\}\right) \lor
                                                                (P_{j0s}+P_{j1s}+P_{j2s}\in\{0,2\}),\hbox{~~~for all $j\in[k]$}.
\end{equation}

\para{Step 2:}
The next step in our \RZQ\ refutation is  showing how to obtain the final contradiction, given the collection of formulas (\ref{eq:011}), for all $s\in[t]$. This is again by counting: we know that for every truth assignment $\overline Z$,  each $Y^{(1)},\ldots,Y^{(t)}$ must contribute at least one clause from $\overline X$ that
is unsatisfiable as 3XOR under $\overline Z$.
We can view this as a mapping $g:[t]\to [m]$ from $Y^{(1)},\ldots,Y^{(t)}$ to the $m$ clauses in $\overline X$, such that $g(i)=j$ means that $Y^{(i)}$ contributes the clause $j$ in $\overline X$ that is unsatisfiable under $\overline Z$ as 3XOR.
The mapping $g$ is not 1-to-1, but $d$-to-1, because every clause of $\overline X$ can appear at most $d$ times in $Y^{(1)},\ldots,Y^{(s)}$.
%
Our \RZQ\ refutation proceeds as follows.

By assumption we have $\sum_{i=1}^m u_i \in \{m-\lceil t/d\rceil +1,\ldots,m\}$, meaning that the number of clauses in $\overline X$ that are satisfied as 3XOR under the assignment $\overline Z$ is at least $m-\lceil t/d\rceil+1$. Also, by the axioms in our formula, for all $i\in[m]$ we can prove that $u_i=1$ implies that $U_i\in\{1,3\}$; namely that the number of true literals in the $i$th clause of $\overline X$ is $1$ or $3$.

For any $s\in[t]$, we can think of $\overline Y^{(s)}$ as a mapping $f^{(s)}:[k]\to [m]$ that maps the $k$ clauses in $\overline Y^{(s)}$ to the clauses in $\overline X$. Then, $y^{(s)}_{ij}=1$ means that $f^{(s)}(i)=j$. Thus, $u_i\cd y^{(s)}_{ji} =1 $ means that the $j$th clause in $\overline Y^{(s)} $ is the $i$th clause in $\overline X$ and that the $i$th clause in $\overline X$ is satisfiable as 3XOR under $\overline Z$.

Now, it is possible to show that for any $s\in[t]$, $i\in[m]$ and $j\in[k]$, there is a proof of the following line:
\begin{equation}\label{eq:041}
\left(u_i\cd y_{ji}^{(s)}=0\right )\lor \left(P_{j0s}+P_{j1s}+P_{j2s}\in\{1,3\}\right)
\end{equation}
which states that if the $i$th clause in $\overline X$ is satisfied as 3XOR under the assignment $\overline Z$ and the $j$th clause in $\overline Y^{(s)}$ maps to the $i$th clause in $\overline X$, then the $j$th clause in $\overline Y^{(s)} $ is satisfied as 3XOR under $\overline Z$.

Informally, the proof of (\ref{eq:041}) is explained as follows: the term $P_{j0s}+P_{j1s}+P_{j2s}$ can be seen as the addition, denoted $\mathcal S$,  of all inner products of the $j$th row of $\overline Y^{(s)}$ with the columns of $\overline X$ (for each $h=0,1,2$ we can consider the column of $\overline X$ restricted to the $h\cd i$ rows only ($i\in[m]$), and so a row of $\overline Y^{(s)}$ which is of length $m$ can have an inner product with such a column of length $m$ in $\overline X$). Because we assume that $y_{ij}^{(s)}=1$, only the $i$th coordinate in the $j$th row of $\overline Y^{(s)} $ is $1$ (and all the other entries  in this row are $0$, by our axioms). Thus, $\mathcal S$ equals in fact a single column from $\overline X$; and this single column is precisely $U_i$.

From (\ref{eq:041}) and (\ref{eq:011}) we can derive, for any  $s\in[t]$ and any $i\in[m]$:
\begin{equation}\label{eq:015}
\bigvee_{j\in[k],i\in[m]}\left(y_{ji}^{(s)}\cd (1-u_i)=1\right),
\end{equation}
stating that for some $j\in[k],i\in[m]$, the $j$th clause in $\overline Y^{(s)}$ is the $i$th clause in $\overline X$ and the $i$th clause in  $\overline X$ is not satisfied as 3XOR under $\overline Z$.

Now, from (\ref{eq:015}) and axioms \ref{it:only-d-joins} in the 3XOR principle formulas, stating that $g:[t]\to[m]$ is $d$-to-$1$, we can obtain that the number of $u_i$'s that are true is no more than $m-\lceil t/d\rceil$, that is, $\sum_{i\in[m]}u_i\in\{0,\ldots,m-\lceil t/d\rceil \}$, contradicting the axiom $\sum_{j=1}^m u_j\in\{m-\lceil t/d\rceil+1,\ldots,m\}$. The formal proofs of this in \RZQ\ is shown in the following lemma:

\begin{lemma}\label{lem:final-rquad-proof}
There are polynomial-size \RZQ\ refutations of (\ref{eq:015}) and the axioms in parts  \ref{it:only-d-joins} and \ref{it:u-sum-i} in the 3XOR principle.
\end{lemma}

\begin{proof}
First sum all axioms (\ref{it:only-d-joins}) to obtain:
\begin{equation}\label{eq:031}
                \sum_{j\in[k], s\in[t]\atop r\in[m]} y^{(s)}_{jr}\in\{0,1,\ldots,d\cd m\}.
\end{equation}
From  (\ref{eq:015}) we can obtain:
\[
                \sum_{j\in[k],r\in[m]} y^{(s)}_{jr}\cd(1-u_i)\in
                                \{ 1,2,\ldots,k\cd m\}, \hbox{~~~~for every $s\in[t]$}. \]
And by summing this for all $s\in[t]$ and $i\in[m]$, we get:
\begin{align}\notag
         \sum_{i\in[m]}\sum_{j\in[k],r\in[m]\atop s\in[t]} y^{(s)}_{jr}\cd(1-u_i) & =
\sum_{i\in[m]} (1-u_i)\cd \sum_{j\in[k],r\in[m]\atop s\in[t]} y^{(s)}_{jr} \\ \label{eq:035} & \in
                                \{t\cd m, t\cd m+1,\ldots,t\cd k\cd m^2\}.
\end{align}
From the axiom in part (\ref{it:u-sum-i}) in the 3XOR principle $\sum_{j=1}^m u_j\in\{m-\lceil t/d\rceil+1,\ldots,m\}$ we can obtain easily
\[
                \sum_{i\in[m]} (1-u_i)\in\{0,1,\ldots,\lceil t/d\rceil-1\}.
\]
From this and (\ref{eq:031}) we get, via Lemma \ref{lem:r-quad-count} proved below, the following:

\[
                 \sum_{i\in[m]} (1-u_i)\cd \sum_{j\in[k], s\in[t]\atop r\in[m]} y^{(s)}_{jr}\in\left\{0,1,\ldots,d\cd m \cd (\lceil t/d\rceil-1 )\right\}.
\]
Since $d\cd m \cd (\lceil t/d\rceil-1 )< d\cd m\cd \lceil t/d\rceil\le m\cd t$, we obtain a contradiction with (\ref{eq:035}), which finishes the refutation.
\end{proof}

It remains to prove Lemma \ref{lem:r-quad-count}, which was used in the above proof:

\begin{lemma}\label{lem:r-quad-count}
Let $\sum_{i\in I} x_i\in\{0,1,\ldots,n\}$ and $\sum_{j\in J} y_j\in\{0,1,\ldots,m\}$ be disjunctions of linear equations, both of size at most $s$. Given these two disjunctions  we can  prove in \RZQ\ with a polynomial-size in $s$   proof, the following:\begin{equation}\label{eq:033}
\sum_{i\in I}x_i\cd \sum_{j\in J} y_j\in\{0,1,\ldots,m\cd n\}.
\end{equation}
\end{lemma}
\begin{proof}
We can reason in a case-by-case manner as follows (see \cite{RT07} on how to carry out informal case-analysis reasoning inside \RL): assume that $\sum_{j\in J}y_j=a$, for $a\in\{0,1,\ldots,m\}$. We wish to show that $x_1\cd\sum_{j\in J}y_j=ax_1$. If $x_1=0$ then $x_1\cd\sum_{j\in J}y_j=0=ax_1$. Otherwise, $x_1=1$. Then, $x_1\cd \sum_{j\in J}y_j=\sum_{j\in J}y_j=a=ax_1$. Since we have the axiom $(x_1=0)\lor(x_1=1)$ we conclude that     $x_1\cd \sum_{j\in J}y_j=a x_1.$
In a similar way we can derive for all $i\in I$:
\begin{equation}\label{eq:037}
x_i\cd \sum_{j\in J}y_j=ax_i.
\end{equation}
And by adding (\ref{eq:037}) for all $i\in I$ we obtain:
\[
\sum_{i\in I} x_i\cd \sum_{j\in J}y_j= a\cd\sum_{i\in I} x_i.
\]
Now using the axiom $\sum_{i\in I} x_i\in\{0,1,\ldots,n\},$ we get
\begin{equation}\label{eq:039}
\sum_{i\in I} x_i\cd \sum_{j\in J}y_j\in\{0,a,2a,\ldots,n\cd a\}.
\end{equation}
Recall that (\ref{eq:039}) was obtained under the assumption that $\sum_{j\in J}y_j =a$. This means that if we have the axiom $\sum_{j\in J}y_j \in\{0,1,\ldots,m\}$, we can obtain:
\begin{equation*}
\sum_{i\in I} x_i\cd \sum_{j\in J}y_j\in\left\{b\cd c \;|\; b\in\{0,1,..,n\} \hbox{ and } c\in\{0,1,\ldots,m\}\right\}=\{0,1,\ldots,n\cd m\}.
\end{equation*}
\end{proof}

Note that the proof of Lemma \ref{lem:r-quad-count} would also work if instead of the sums $\sum_{i\in I}x_i$ or $\sum_{j\in J}y_j$ we have $\sum_{i\in I}b_ix_i$ or $\sum_{j\in J}c_jy_j$, for integers $b_i, c_j$.

\section{Reduction to weak automatizability of \RZ}
Here we show that \RZ\ is weakly automatizable if and only if \RZQ\ is weakly automatizable.

To show that \RZ\ is  weakly automatizable iff \RZQ\ is weakly automatizable we use a similar idea to Pudl\'ak \cite{Pud03-tcs}. Namely, we show that the \emph{canonical pair} of \RZQ\ is polynomially reducible to the canonical pair of \RZ.

\begin{definition}[\cite{Razb94}]
The \emph{canonical pair} of a refutation system $\mathcal P$ is the disjoint \NP-pair, whose first \NP\ language consists of all pairs $(\tau,1^m)$ where $\tau$ is an unsatisfiable formula that has a $\mathcal P$-refutation of size at most $m$, and whose second \NP\ language is the set of pairs $(\mu,1^m)$ where $\mu$ is a  satisfiable formula and $m$ is some natural number.
\end{definition}

We say that a canonical pair $(A,B)$ of a refutation system $\mathcal  P'$ is \emph{polynomially reducible} to the canonical pair $(A',B')$ of another refutation system  $\mathcal P$ if there is a polynomial-time computable function $f$ such that for all $x$ it holds that $x\in A \iff f(x)\in A'$ and $x\in B \iff f(x)\in B'$. A simple corollary of the above definitions is the following:

\begin{proposition}[\cite{Pud03-tcs}]
If the canonical pair of $\mathcal P'$ is polynomially reducible to the canonical pair of $\mathcal P $ then $\mathcal P'$ is weakly automatizable if $\mathcal P$ is weakly automatizable.
\end{proposition}


In view of this proposition, and since \RZQ\ clearly polynomially simulates \RZ\ (as an extension of it),   it remains to show the following:

\begin{proposition}
The canonical pair of \RZQ\ is polynomially reducible to the canonical pair of \RZ.
\end{proposition}

\begin{proof}(Sketch)
Similar to \cite{Pud03-tcs}, the idea is to encode a product of any two variables $x_i\cd x_j$ as a new single formal variable $x_{ij}$. Thus, the reduction sends all pairs $(\tau,1^m)$ to the pair
$(\tau',1^{{\rm poly}(m)})$, where $\tau'$ is obtained from $\tau$ by adding the axioms that force all new variables $x_{ij}$ to encode the product $x_i\cd x_j$, as shown in Section \ref{sec:encoding}. \end{proof}

\setcounter{tmpc}{\thetheorem}
\setcounter{theorem}{\thecD}
\begin{corollary}
\RZQ\ is weakly automatizable iff\, \RZ\ is weakly automatizable.
\end{corollary}
\setcounter{theorem}{\thetmpc}

Since \RZQ\ admits polynomial-size refutations of the 3XOR principle, and since weak automatizability entails feasible interpolation, we get a reduction of the problem  of determinizing Feige et al.~nondeterministic refutation algorithm to the problem of establishing weak automatizability of \RZ:
\setcounter{tmpc}{\thetheorem}
\setcounter{theorem}{\thecE}
\begin{corollary}
If \,\RZ\ is weakly automatizable then there is a deterministic refutation algorithm for random 3CNFs with $\Omega(n^{1.4})$ clauses.
\end{corollary}
\setcounter{theorem}{\thetmpc}

\section*{Acknowledgments}
I wish to thank Jan Kraj\'{i}\v{c}ek for useful comments related to this work and Albert Atserias and Neil Thapen for useful related discussions. Thanks also to the anonymous reviewers of this paper who helped much in improving the exposition of this work.

\bibliographystyle{plain}
\bibliography{PrfCmplx-Bakoma}
\end{document}